\documentclass[table]{ccjnl}
\usepackage{lipsum,amsmath}
\usepackage{cuted}
\graphicspath{{figures/}}

\usepackage{subfig}

\title{Linear Network Coding Based Fast Data Synchronization for Wireless Ad Hoc Networks with Controlled Topology}

\author{Die Hu\inst{1,2}, Xuejun Zhu\inst{2}, Min Gong\inst{2}, Shaoshi Yang\inst{1,3,*} \corinfo{shaoshi.yang@bupt.edu.cn}
}
\receiveddate{Jan. 17, 2021}
\reviseddate{}
\Editor{}

\address[1]{School of Information and Communication Engineering, Beijing University of Posts and Telecommunications, Beijing 100876, China}
\address[2]{China Academy of Launch Vehicle Technology, Beijing 100076, China}
\address[3]{Key Laboratory of Universal Wireless Communications, Ministry of Education, Beijing 100876, China}

\begin{document}

\maketitle

\begin{abstract}
Fast data synchronization in wireless ad hoc networks is a challenging and critical problem. It is fundamental for efficient information fusion, control and decision in distributed systems. Previously, distributed data synchronization was mainly studied in the latency-tolerant distributed databases, or assuming the general model of wireless ad hoc networks. In this paper, we propose a pair of linear network coding (NC) and all-to-all broadcast based fast data synchronization algorithms for wireless ad hoc networks whose topology is under operator's control. We consider both data block selection and transmitting node selection for exploiting the benefits of NC. Instead of using the store-and-forward protocol as in the conventional uncoded approach, a compute-and-forward protocol is used in our scheme, which improves the transmission efficiency. The performance of the proposed algorithms is studied under different values of network size, network connection degree, and per-hop packet error rate. Simulation results demonstrate that our algorithms significantly reduce the times slots used for data synchronization compared with the baseline that does not use NC.
\end{abstract}
\section{Introduction}
\label{Introduction}
Wireless ad hoc networks, such as the networked sensors, robots and unmanned aerial vehicles (UAVs), constitute a distributed, flexible and cooperative information-sharing system \cite{Meng_2015:SkyStitchAC,Balasingam2015:Robust_collaborative,Fei2017:Wireless_Sensor}. 
Fast data synchronization among network nodes is important for wireless ad hoc networks, since it is expected to provide essential information reliably for high-layer real-time decision and control algorithms. Unfortunately, in general this is a great challenge, because in many cases the network topologies and the wireless channels between nodes are highly dynamic and random \cite{Xie2009:Multi-Channel,Lv2016:Localization}. 

Owing to the openness of wireless channels, all-to-all broadcast has the potential to serve as a highly efficient approach for achieving fast data synchronization in wireless ad hoc  networks. However, the transmission efficiency of all-to-all broadcast is still constrained by the conventional store-and-forward protocol, which does not take advantage of more sophisticated data processing and thus limits the distributed system's efficiency of sensing and reacting to the environment. The network coding (NC) technique \cite{Ahlswede2000:Network_information,
Bassoli2013:netwoek_survey} offers an attractive approach for data synchronization in distributed systems, since it is capable of reducing the total time cost of data synchronization by exploiting both the broadcast property of wireless channels and the XOR operation in NC, thus increasing the effective data transmission rate \cite{Katti2008:Wireless_Network_Coding,Cao2013:Two-Way_Relay}.

Regarding the related work, in \cite{Asterjadhi2010:coding-based_protocols} the authors proposed a proactive NC scheme for all-to-all broadcast while using the random access schemes of IEEE 802.11 and considering several different network topologies. In \cite{Duan2015:all-to-all_broadcasting} an all-to-all broadcast protocol was designed for wireless ad hoc networks that use directional antennas, but NC was not considered. In \cite{Ma2013:Reliability_all-to-all} the reliability of a random neighbor NC  based all-to-all broadcast approach was analyzed. In \cite{Borkotoky2019:All-to-All_Broadcast} an adaptive transmission protocol suite integrated with the random linear NC was proposed, where the modulation and channel coding parameters are adaptively selected for each packet to improve the packet loss performance.  

However, all the above contributions are designed for general ad hoc networks, while ignoring the unique properties of specific systems. There indeed exist some particular scenarios, where the original distributed data synchronization problem can become less challenging or enjoy more benefits. For instance, compared with the data integrity and security, the latency requirement is less stringent for the data synchronization in distributed databases \cite{Choi_2009:Synchronization, Son1988:Data_Management}. 
Additionally, more benefits can be gleaned in the scenario where the network topology is under the operator's control. Topology control is important and practical for wireless ad hoc networks, as it is beneficial for reducing energy consumption (thus extending the network lifetime) and radio interference (thus increasing the network communication capacity) \cite{Paolo2005:Topology_Control, Chen2002:Topology_Maintenance, Rajmohan2002:Topology_Control_Survey}. By using topology control, a more stable and convenient graph representation of the network can be obtained. Wireless ad hoc networks with controlled topology have found applications in many areas, such as the flight formation in military operations, the truck platooning in Internet of Vehicles (IoV) as highlighted by 5G, and the low earth orbit (LEO) satellite constellations etc. 

In this paper we propose a pair of linear NC and all-to-all broadcast based fast data synchronization algorithms for wireless ad hoc networks with controlled topology.  We study the performance of the proposed algorithms with a large number of randomly generated network topology samples. It is shown that the average gain in terms of the time-slot usage reduction upon applying the proposed algorithms can be over five times compared with the time-division multiple-access (TDMA) based all-to-all broadcast algorithm that does not use NC. Furthermore, this substantial gain is achievable in a wide range of network topologies, in particular for the topologies that have a low or medium degree of network connectivity.  

\section{System Model}\label{Sec:intro}
As shown in Figure~\ref{fig:System Model}, we consider an ad-hoc network that consists of $N$ nodes connected in any arbitrary topology. We assume that for each node that carries out linear NC, the corresponding decoding is performed by its immediate neighbouring nodes.  For each hop, a single action of transmitting-and-receiving occupies one time slot. The data synchronization throughout the network is achieved as follows. The node $n$ broadcasts its own data block $p_n$, whose original or processed copy is then disseminated to the other $N-1$ nodes by relaying through their respective neighbouring nodes, where $n=1, 2, \cdots, N$. The number of neighbouring nodes of node $n$ and the number of data blocks stored at node $n$ are denoted as $m_n$ and $d_n$, respectively. It is required that all the nodes in the network obtain the set of data blocks ${\mathbb{A}} = \left\{{{p_1},{p_2},\cdots,{p_N}} \right\}$ as fast as possible, which means that ${d_n}$ must equal 
$N$ after the data synchronization is achieved throughout the network, i.e., $\max{\{d_n\}} = N$, $\forall n= 1, 2, \cdots, N$. 

In a given time slot the node $n$ shares its data blocks with its neighbouring nodes $n_k$, where $1 \le k \le m_n$, $1 \le n_k \le N$ and $n_k \neq n$. To this end, we define the following sets:

\begin{itemize}
\item ${{\mathbb{A}}_{n}}{\rm{ = }}\left\{ {{p_{1}},{p_{2}},\cdots,{p_{{d_n}}}} \right\}$: The data blocks stored at node $n$.

\item $\bar{\mathbb{A}}_{n_k} = \{ \left. {{p_j}} \right|{p_j} \in {\mathbb{A}}, {p_j} \notin {{\mathbb{A}}_{n_k}}, j = 1, 2, \cdots, N\}$: The data blocks that node ${n_k}$ has not obtained.
\end{itemize}

Then $\bar{\mathbb{A}}_{n_k} \cap {\mathbb{A}}_n$ represents the data blocks that node $n_k$ can obtain from node $n$.

\begin{figure}[t]
\begin{center}
\includegraphics[angle = 0,width = 2.3in]{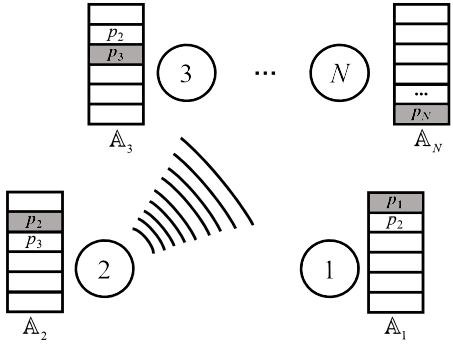}
\caption{The system model of the wireless ad hoc network considered, where the shaded rectangle $p_n$ denotes the data block stored at node $n$ in the initial stage. } 
\label{fig:System Model}
\end{center}
\end{figure}

\section{The Proposed Data Synchronization Algorithms}
\label{Sec:Data_Sync}

Below we describe the proposed data synchronization algorithms in detail, whose flowchart is shown in Figure~\ref{fig:Synchronous strategy}.  

\textbf{Step 0:} This is the initialization stage, where each node carries out data acquisition by monitoring its own state or the state of the environment. Note that it is possible for a single node to have multiple data blocks.  

\textbf{Step 1:}
Each node broadcasts the data blocks, which are stored on it and encapsulated into packets, to its neighbours in a TDMA manner. 

\textbf{Step 2:} 
If a packet is received correctly and network-coded, it is fed into the network-decoding module deployed on the node that receives the packet, and the data blocks extracted from the decoder are stored on the node. If the packet is not received correctly, each node updates its data block storage status.  Otherwise, the packet is directly stored on the node in the format of data blocks, and each node updates its data block storage status accordingly. 

\textbf{Step 3:} Check whether all the data blocks are synchronized in all nodes. This is achieved at each node by examining if $d_n = N$. If yes, the algorithms terminate. Otherwise, select an appropriate transmitting node that contributes on average the maximum innovative data to each neighbouring node according to \eqref{eq:node_selection} (\textit{optional}, used in Alg. 2, by examining the feedback from the subsequent data block selection (DBS), as elaborated in Sec.~\ref{section:selection}), and an appropriate set of data blocks that are innovative and decodable to the maximum number of the neighbouring nodes of the selected node according to \eqref{eq:data_block_selection_criterion} (\textit{mandatory}, used in both Alg. 1 and Alg. 2, by exploiting the output of the decoder in Step 2, as elaborated in Sec.~\ref{subsec:re-encoding}). These selected data blocks are stored in their host node to participate in the subsequent network-encoding operation.  

\textbf{Step 4:}
Carry out network-encoding operation with respect to the data blocks selected, and jump to Step 1, where the above selected node broadcasts the network-coded data blocks to their respective neighbours. This process is repeated until all the data blocks on all the nodes are synchronized. 

Note that the DBS results are fed back to serve as an input for the node selection (NS) module, according to the outputs of the decoding module introduced in Sec.~\ref{subsec:decoding}. 

\begin{figure*}[t]
\begin{center}
\includegraphics[angle = 0,width = 6 in]{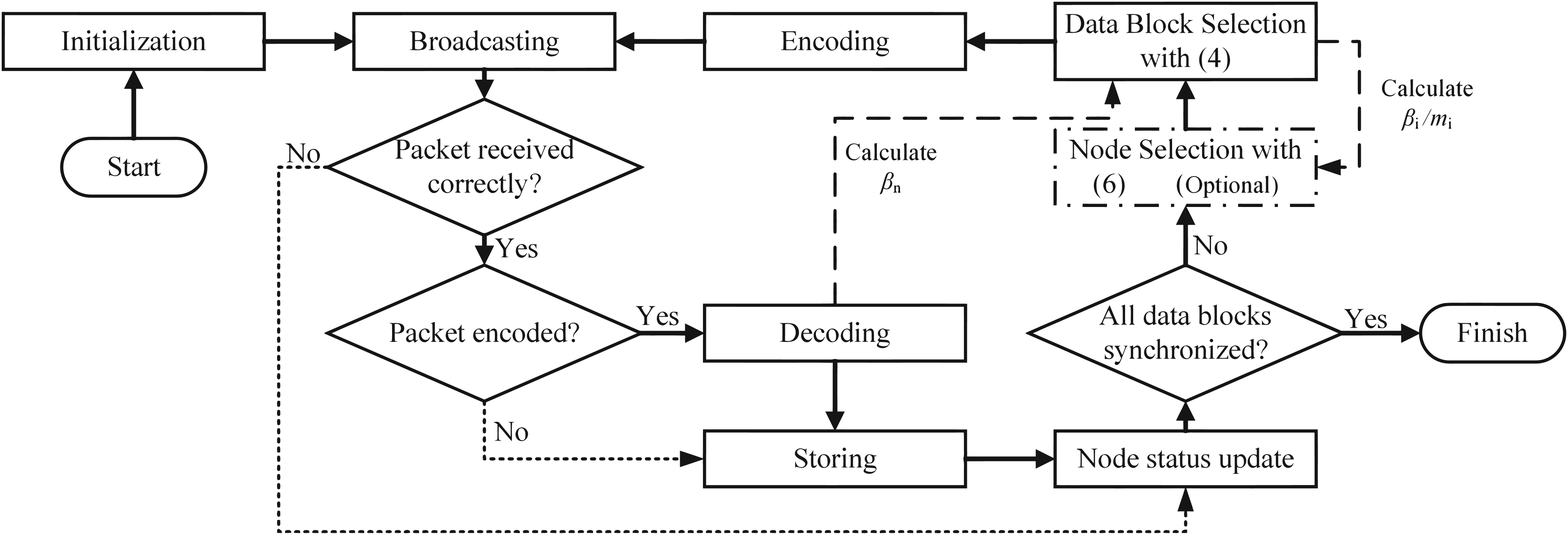}
\caption{The flowchart of the proposed data synchronization algorithms, where the node selection is only invoked in Alg. 2.} 
\label{fig:Synchronous strategy}
\end{center}
\end{figure*}

\subsection{Decoding}
\label{subsec:decoding}
For each TX-RX pair, the receiver is anticipated to decode the packets sent by the transmitter and then keep the data blocks that are absent from this receiver previously. In other words, each receiver is only interested in the data blocks that are innovative to itself. To clarify our design philosophy, the following theorem is given. 

\begin{theorem}
For packets generated by linear NC, the receiver is capable of decoding the packet if at most one component data block is unknown.
\end{theorem}

\begin{proof}
Let us use the mathematical induction method to prove the theorem. 
Firstly, we have the following proposition:  the component data block $x$ can be extracted from packet $P = x \oplus {a_1} \oplus {a_2} \oplus \cdots {a_l} \cdots \oplus {a_L}$  provided that all the component data blocks $a_l, l = 1, 2, \cdots, L$ are known, where $L \in {\mathbb{Z}_+} $ is a positive integer.

1) When $L=1$,  the proposition is obviously true.

2) Assume that the above proposition is true when $L = r$, with $r \in {\mathbb{Z}_+}$ being an arbitrarily given positive integer. Then $x$ can be solved from $P = x \oplus a_1 \oplus a_2 \oplus \cdots {a_l} \cdots \oplus a_r$ if and only if $a_l$ is known, $\forall l \in \{1,2, \cdots, r\}$.

3) When $L=r+1$, the packet $P$ is written as
\begin{equation}
\label{eq:proof}
P=x \oplus a_1 \oplus a_2 \oplus \cdots {a_l} \cdots \oplus a_r \oplus a_{r+1}.
\end{equation}

Then if $\forall l \in \{1,2, \cdots, r+1\}$, where $a_l$ is unknown, $x$ is obviously  unsolvable. 
Let us assume that at least one $a_l$ is known. For convenience, assume that $a_{r+1}$ is known. Then Eq. \eqref{eq:proof} is rewritten as  

\begin{equation}
P=P' \oplus a_{r+1},
\end{equation}
where $P'=x \oplus a_1 \oplus a_2 \oplus ... \oplus a_r$ and it is decodable. Furthermore, to decode the data block $x$,  $\forall a_l, l=1,2,...,r$ must be known because of the above statement 2).  

Thus the proof has been established.  
\end{proof}

Therefore, upon receiving the packet $P$, the node $n_{k}$ is able to obtain innovative data block from node $n$ if
\begin{equation}
\label{eq:decoding}
|{\mathbb{A}_n} \cap {\bar{\mathbb{A}}_{n_k}}| = 1.
\end{equation}

\subsection{The DBS operation and encoding}\label{subsec:re-encoding}

To maximize the gain of all-to-all broadcast,  the network encoding operation is expected to select the particular data blocks that enable as many neighbouring receivers as possible to decode the packet transmitted  and each of these receivers must obtain innovative data block from the decoding.
In other words, the network encoding should operate on the data blocks that are the solutions to the following optimization problem  
\begin{equation}\label{eq:data_block_selection_criterion}
\max_{{\tilde{\mathbb{A}}_n} \subseteq {\mathbb{A}_n}} {\beta _{n}},
\end{equation}
where $\tilde{\mathbb{A}}_n$ denotes the set of data blocks selected on the node $n$ and $ {\beta _n} $ is defined as the particular number of the neighbouring nodes that can decode and obtain innovative data from node $n$: 
\begin{equation}
\label{eq:benefit}
{\beta _n} = \sum\limits_j^{m_n} {{\varepsilon _{nj}}}. 
\end{equation}
More specifically, ${\varepsilon _{nj}}{\rm{ = }}\left\{ {\begin{array}{*{20}{c}}
1&{ |\tilde{{\mathbb{A}}}_n \cap {\bar{\mathbb{A}}_{n_j}}| = 1 }\\
0&{\textrm{otherwise}}
\end{array}} \right.$, and $j$ represents the index of receiving nodes. 
\eqref{eq:data_block_selection_criterion} is a combinatorial optimization problem, which can be solved by any established search method $f(\cdot)$ over all the subsets of ${\mathbb{A}_n}$.   

To be more efficient, the node $n$ can select data blocks to transmit from the set 
$\mathbb{B}=(\bar{\mathbb{A}}_{n_1} \cap {\mathbb{A}}_n) \cup (\bar{\mathbb{A}}_{n_2} \cap {\mathbb{A}}_n) \cup \cdots \cup (\bar{\mathbb{A}}_{n_{m_n}} \cap {\mathbb{A}}_n)$
with the specific approach $f(\cdot)$ introduced above\endnote{In other words, $\tilde {\mathbb A}_n \subseteq \mathbb{A}_n$ in \eqref{eq:data_block_selection_criterion} can be replaced with $\tilde {\mathbb A}_n \subseteq \mathbb{B}$, resulting in a smaller search space.}, and encodes these component data blocks into packet $P$ with the aid of the bitwise XOR operation, i.e., $P = {p_{d_1}} \oplus {p_{d_2}} \oplus \cdots \oplus {p_{d_j}}$, where $ \oplus $ stands for bitwise XOR, $1 \leq d_j \leq N$, $1 \leq j \leq | \tilde{{\mathbb {A}}}_n | $ and the selected data blocks ${p_{d_j}} \in \tilde{\mathbb{A}}_n = f (\mathbb{B})$. 

\subsection{The NS operation}
\label{section:selection}
For a given network topology, it is important to determine which node should broadcast its data to the neighbours at a particular time instant. In other words, it is necessary to carry out NS, in addition to DBS. Thus, based on $\beta_n$ we further consider the ratio of the neighbouring nodes that can be helped among all the neighbouring nodes of the selected node per broadcast or time slot. Then, the node is selected according to 
\begin{equation}\label{eq:node_selection}
\max_i {\beta _i}/{m_i},
\end{equation} 
where $i = 1, 2, \cdots, N$. By solving \eqref{eq:node_selection}, the node that enables the largest proportion of neighbour nodes to be capable of obtaining innovative data blocks is selected, thus the transmission efficiency is improved.

\section{Simulation Results and Discussions}\label{Sec:simulation}
\begin{figure*}[t]
\vspace{-0.5cm}
    \centering
	  \subfloat[Node size = 5 nodes]{
       \includegraphics[width=0.32\textwidth]{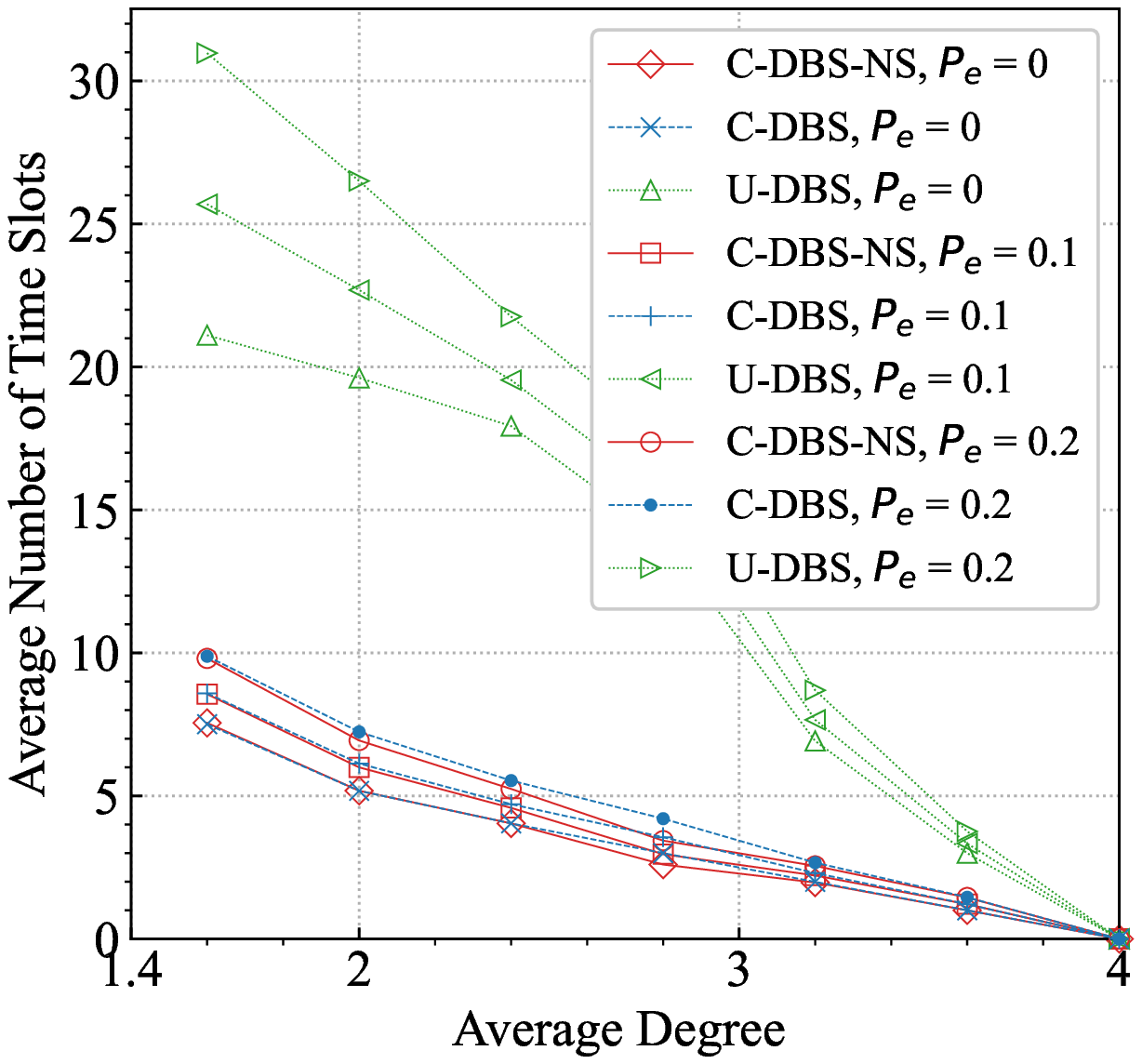}}
    \label{1a}
%
%
	  \subfloat[Node size = 8 nodes]{
        \includegraphics[width=0.32\textwidth]{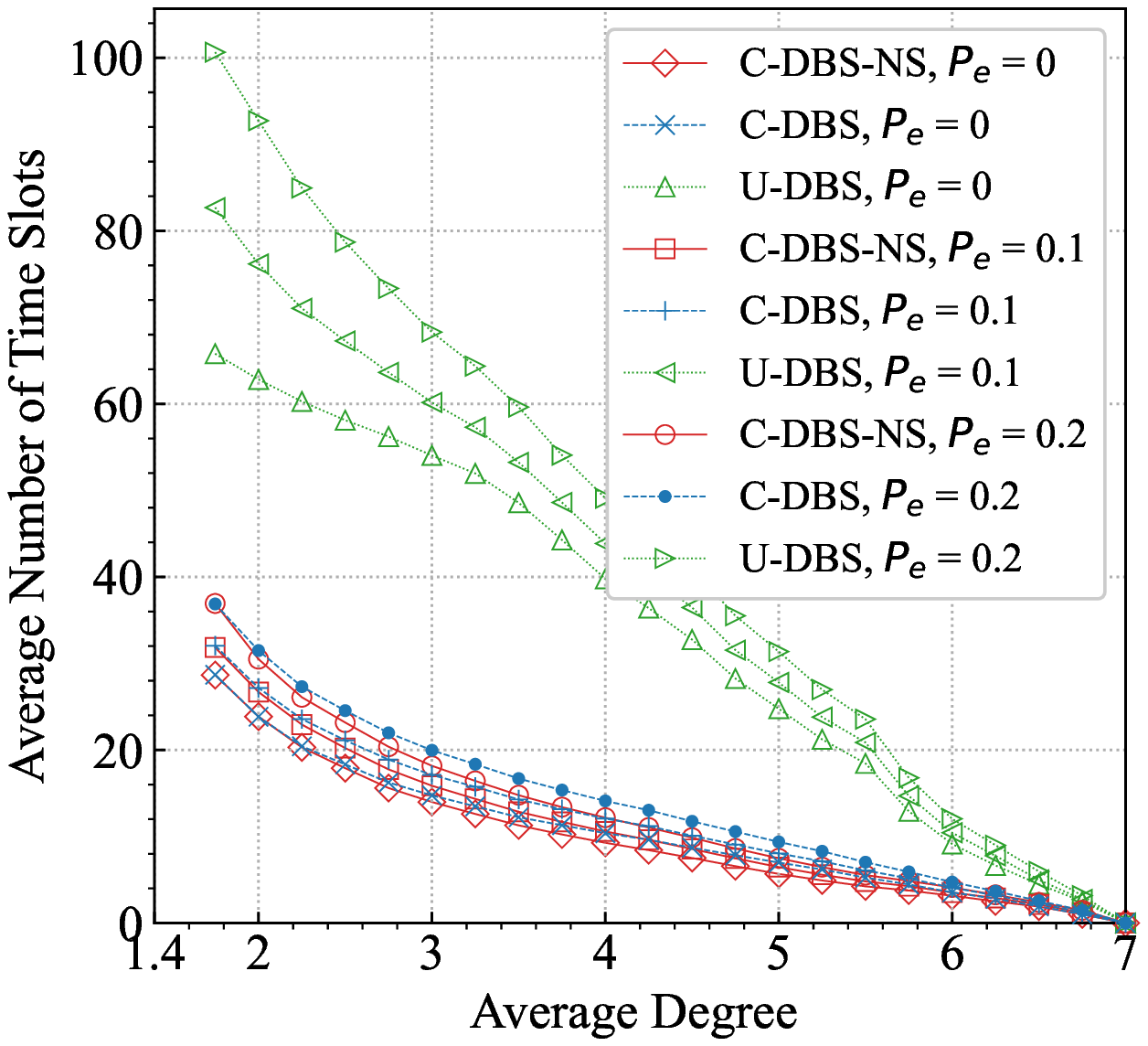}}
    \label{1b}
	  \subfloat[Node size = 11 nodes]{
        \includegraphics[width=0.32\textwidth]{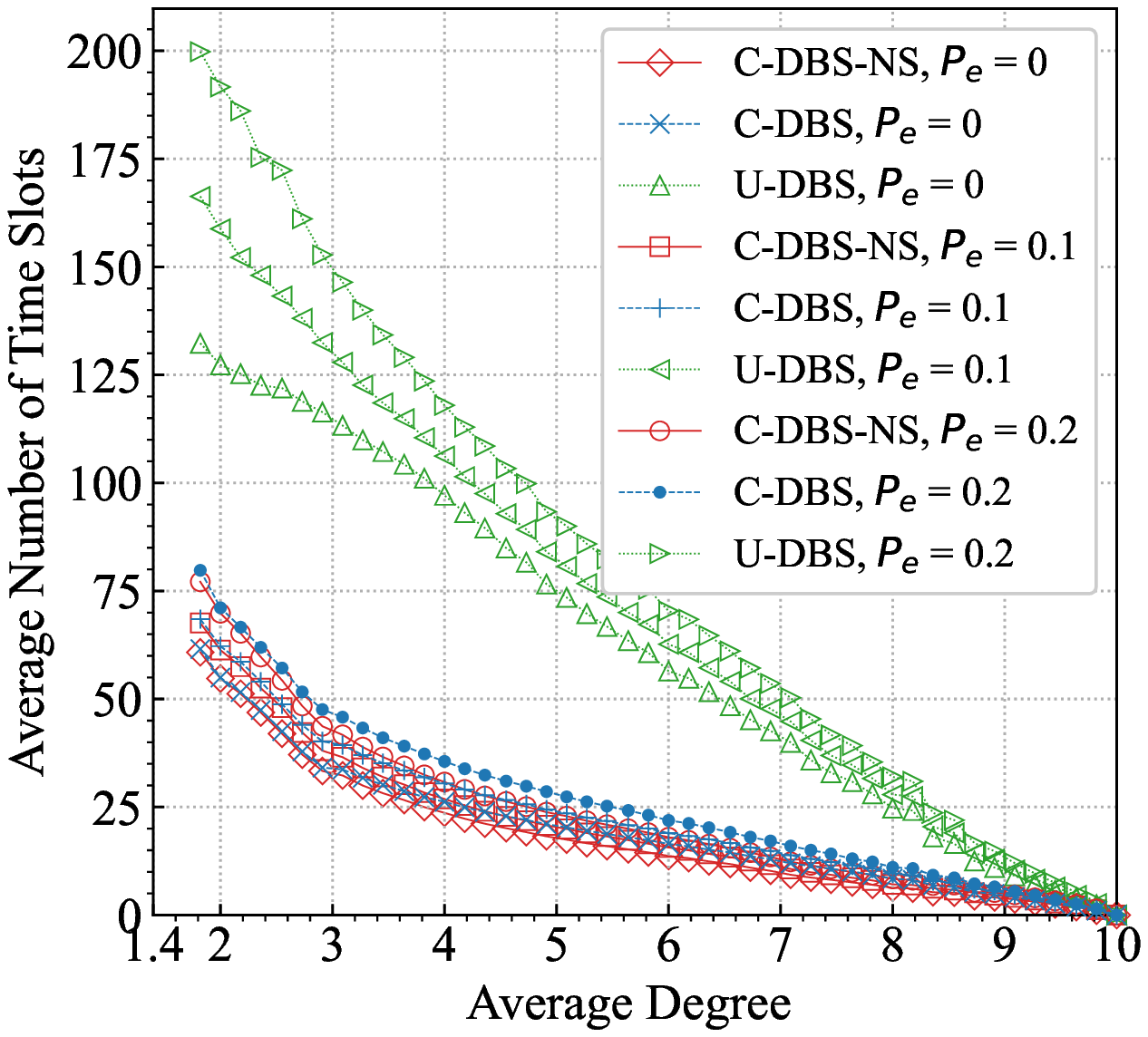}}
    \label{1c}
	  \caption{The average number of time slots used by different schemes versus the average degree of a wireless ad hoc network.}
	  \label{fig:Slot} 
	  \vspace{-0.5cm}
\end{figure*}

In this section, the performance of the proposed all-to-all broadcast based  distributed data synchronization algorithms are demonstrated with Monte Carlo simulations. We consider a wireless ad hoc network, where the number of nodes is from 5 to 11.  For each network size, $10^5$ network samples are randomly generated, and in each network sample the positions of the nodes are generated with uniform distributions, thus guaranteeing a diverse range of network topologies to be examined.   Additionally, we only keep the network samples that can be represented by connected graph, in order to ensure that it is feasible to achieve data synchronization throughout the network. The comparisons are made among three data synchronization schemes:
\begin{enumerate}
\item \textit{Uncoded+DBS (baseline, abbreviated as U-DBS)}: The data synchronization is performed without using NC. In a single transmission cycle, each node broadcasts its individual data in turn following a given order, and each node uses a single time slot. DBS is also invoked by each node in this scheme, using the criterion similar to \eqref{eq:data_block_selection_criterion}, where $\tilde{\mathbb{A}}_n$ changes from a multi-element set to a single-element set. In other words, each node broadcasts the data block that is innovative to the largest number of receivers. This process is repeated in the next cycle, until the data synchronization across the network is achieved using the conventional store-and-forward protocol.   

\item \textit{Coded+DBS (Alg.~1, abbreviated as C-DBS)}: The data synchronization is achieved by employing the DBS based linear NC, as presented in Sec.~\ref{subsec:re-encoding}, while NS is not invoked. The transmission protocol can now be termed as \textit{compute-and-forward}. The remaining operations are the same as the scheme of \textit{U-DBS}. 

\item \textit{Coded+DBS+NS (Alg. 2, abbreviated as C-DBS-NS)}: In this scheme, not only the DBS based linear NC, but also the NS as introduced in Sec.~\ref{section:selection}, is invoked. Thus,  in each time slot, only the node selected broadcasts its data. This process is repeated until the data synchronization across the network is achieved. As such, the concept of \textit{transmission cycle} as used in the schemes of \textit{U-DBS} and \textit{C-DBS} becomes invalid, and the transmission protocol is also the \textit{compute-and-forward}. 
\end{enumerate}

Firstly, by using the concept of \textit{average degree} as defined in graph theory for characterizing the nodes' and network's connectivity\endnote{As the average degree increases, the network connection state is improved. }, in Figure~\ref{fig:Slot} we evaluate the average number of time slots used by the three schemes considered, under different values of network degree, network size and packet error rate $P_e$ of a single hop. 
The average number of time slots $S$ under a particular average degree $d$ is defined as 
\begin{equation}
S{}_{d} = \frac{{\sum\limits_i^{{N_d}} {s{_{id}}} }}{N_d},
\end{equation}
where $s{_{id}}$ denotes the number of time slots used in the data synchronization process for the $i$th network sample under the average degree $d$, and $N_d$ denotes the number of randomly generated network samples, which essentially also represent different network topologies.

\begin{figure*}[t]
    \centering
    \vspace{-0.5cm}
	  \subfloat[Node size = 5 nodes]{
       \includegraphics[width=0.32\textwidth]{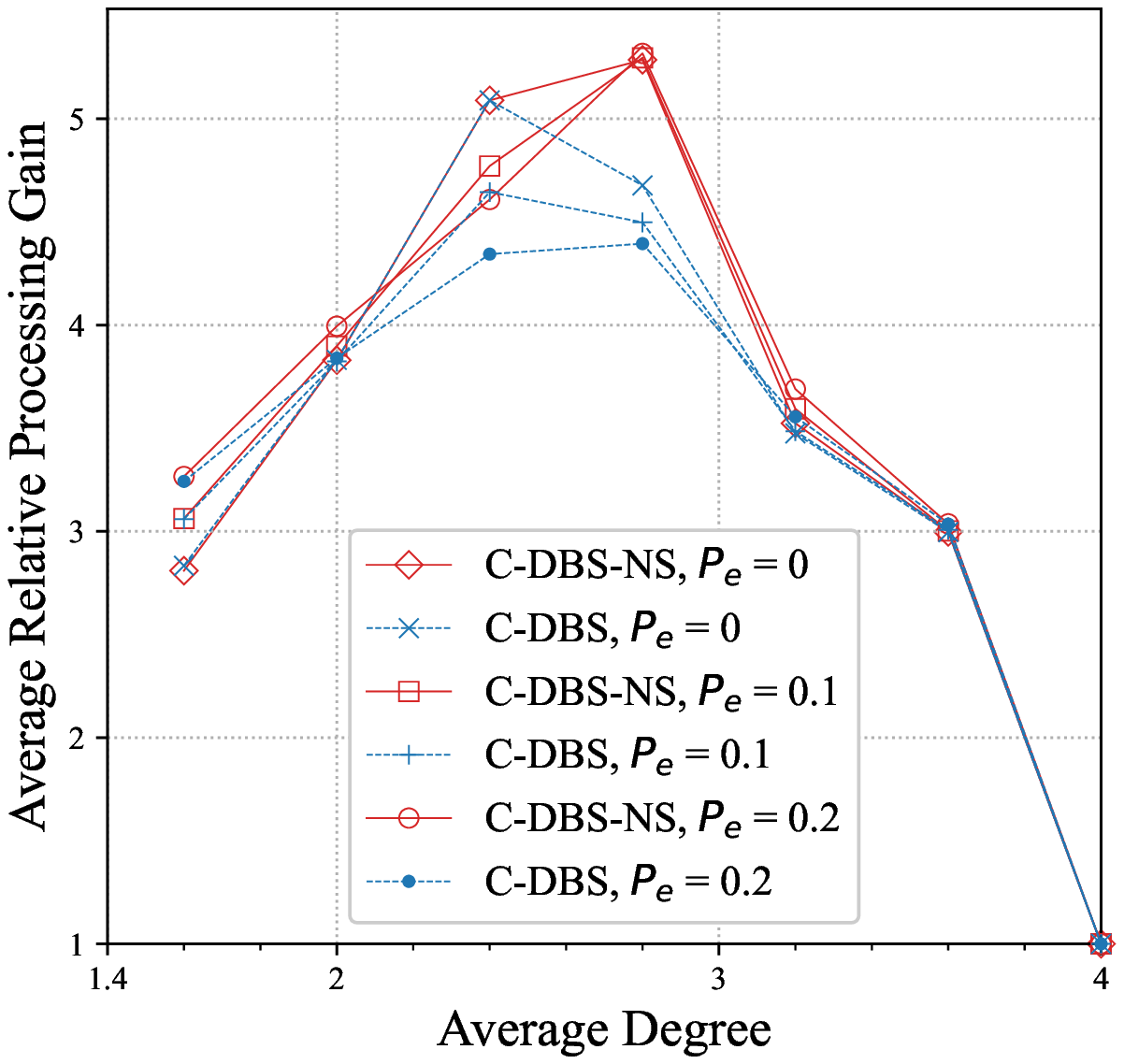}}
    \label{1a}
	  \subfloat[Node size = 8 nodes]{
        \includegraphics[width=0.32\textwidth]{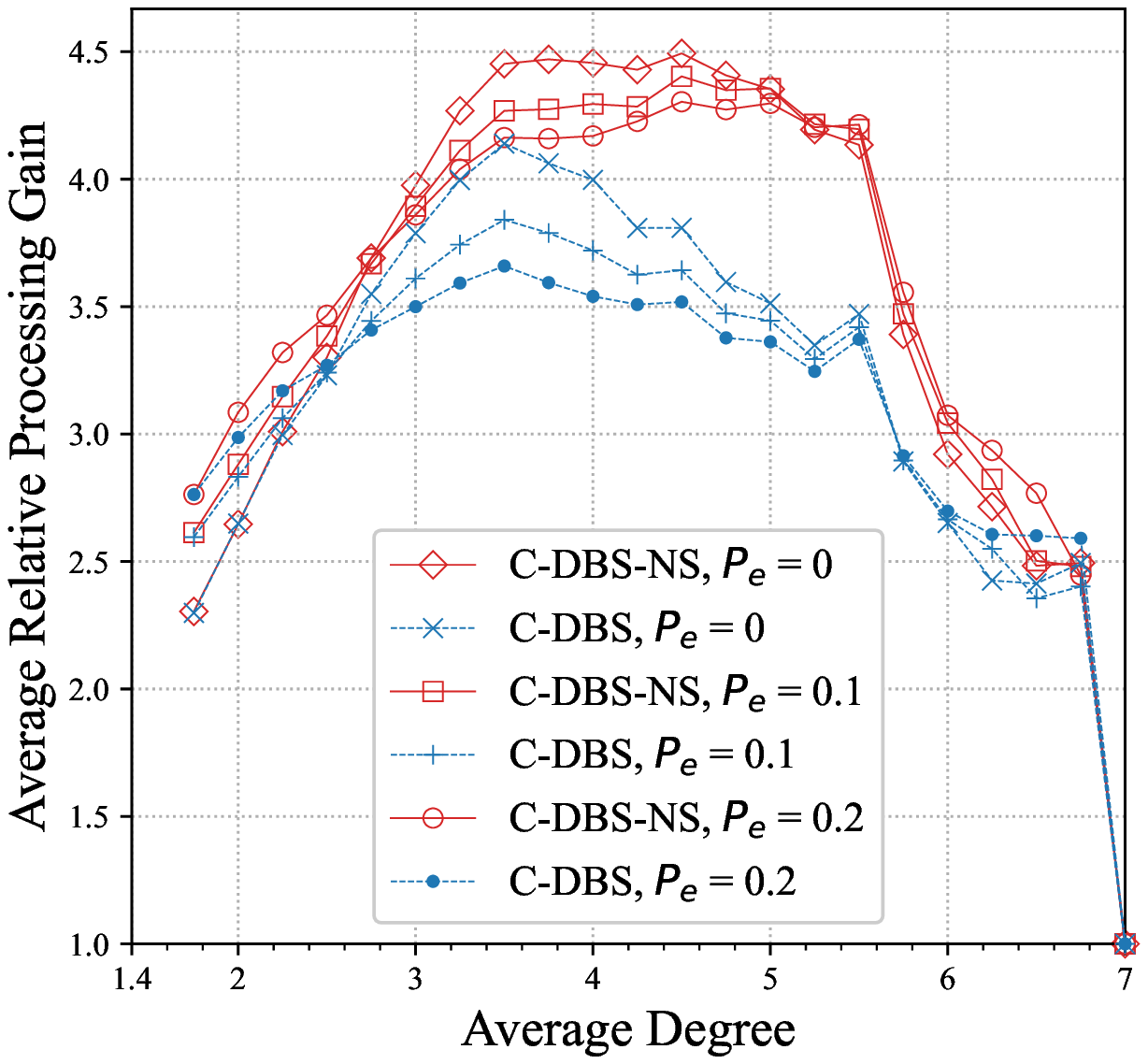}}
    \label{1b}
	  \subfloat[Node size = 11 nodes]{
        \includegraphics[width=0.32\textwidth]{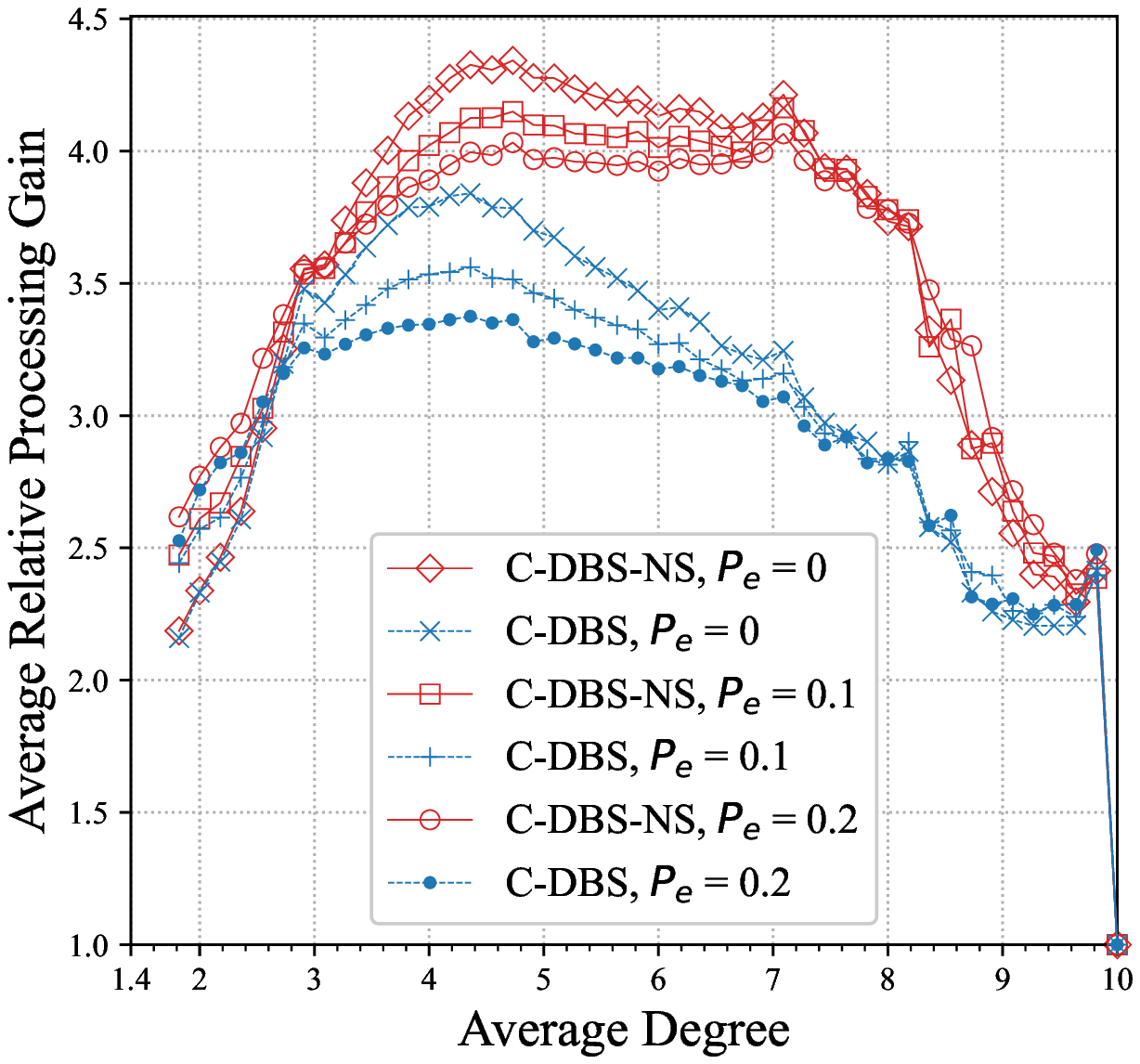}}
    \label{1c}
	  \caption{Average gain versus average degree}
	  \label{fig:Gain} 
	  \vspace{-0.5cm}
\end{figure*}

We can observe that the proposed \textit{C-DBS-NS} scheme (the red curves) and \textit{C-DBS} scheme (the blue curves) both enjoy a significant time slot usage reduction, compared to the \textit{U-DBS} scheme (the green curves) in the weakly or moderately connected networks. Additionally, the \textit{C-DBS-NS} scheme performs best among the three schemes, especially under the moderately connected network topology, but the advantage of the \textit{C-DBS-NS} scheme over the \textit{C-DBS} scheme is limited. Furthermore, the time slot usage increases when we have a larger $P_e$ and a larger network size.



Secondly, the average relative processing gain (RPG) of the proposed two schemes to the baseline \textit{U-DBS} scheme under different values of network degree, network size and packet error rate $P_e$ of a single hop is depicted in Figure~\ref{fig:Gain}. Here the average RPG is defined as
\begin{equation}
{G_d} = \frac{{\sum\limits_i^{{N_d}} {{\hat {s}_{id}}/{{\tilde s}_{id}}} }}{{{N_d}}},
\end{equation} 
where ${\tilde s}_{id}$ and ${\hat s}_{id}$ respectively denote the number of time slots used in 
the data synchronization process of the \textit{U-DBS} scheme and of either the proposed \textit{C-DBS} or \textit{C-DBS-NS} schemes for the $i$th network sample under the average degree $d$. We observe that for either large or small average degree, the average RPG of the \textit{C-DBS} scheme is similar to that of the \textit{C-DBS-NS} scheme, which is a common phenomenon under different values of network size. For medium average degree, the average RPG of the \textit{C-DBS-NS} scheme is significantly higher than that of the \textit{C-DBS} scheme. In other words, for poorly connected networks, the \textit{C-DBS} scheme has a similar performance compared to the most complicated \textit{C-DBS-NS} scheme; for well connected networks, the data  synchronization performance (i.e., synchronization speed) of all the three schemes, including the simplest \textit{U-DBS} scheme, tends to be identical; while for moderately connected networks, the \textit{C-DBS-NS} scheme is shown to have the highest synchronization speed.

Finally, in Figure~\ref{fig:WholeTime} we compare the computational complexity of the three schemes in terms of the average number of FLOPS\endnote{Due to the randomness of the topology and the network size, it is difficult to give analytic results of the computational complexity of the three schemes.} under different numbers of nodes. We see that as the number of nodes rises, the average number of FLOPS of the \textit{C-DBS-NS} scheme grows much faster than that of the \textit{U-DBS} scheme, while the \textit{C-DBS} scheme has a moderately higher computational complexity than the  \textit{U-DBS} scheme. This implies that when choosing appropriate data synchronization algorithm for large-scale wireless ad hoc networks, both the time slot usage and the computational complexity of the algorithm should be considered. It is advised to use the divide-and-conquer approach to reduce the network size to a moderate value.

\begin{figure}[t]
\begin{center}
\includegraphics[angle = 0,width = 3.5in]{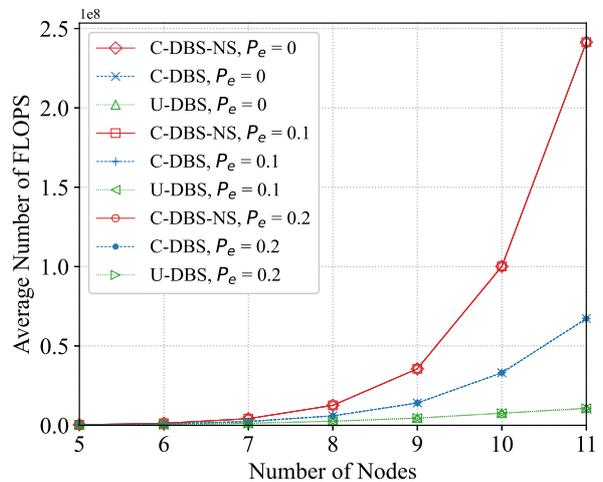}
\caption{Computational complexity versus the network size } 
\label{fig:WholeTime}
\end{center}
\vspace{-0.5cm}
\end{figure}
\section{Conclusion} \label{Sec:conclusion}
In this paper, we have proposed a pair of linear NC and all-to-all broadcast based fast data synchronization algorithms for wireless ad hoc networks that have controlled topology. For better exploiting the benefits of NC, the first algorithm \textit{C-DBS} exploits data block selection, while the second algorithm \textit{C-DBS-NS} exploits both the transmitting node selection and data block selection. Thus, compared with the conventional uncoded approach that uses store-and-forward protocol, a more efficient compute-and-forward protocol is used in our algorithms. We show that \textit{C-DBS-NS} performs best among the three schemes in the moderately connected networks, while in the weakly connected networks, \textit{C-DBS} achieves similarly good performance at a lower computational complexity compared to  \textit{C-DBS-NS}. For well connected networks, however, the three schemes have almost identical performance, hence there is no need to perform NC and the simplest \textit{U-DBS} approach suffices.  

\theendnotes

\section*{ACKNOWLEDGEMENT}
\label{ACKNOWLEDGEMENT}
This work is financially supported by Beijing Municipal Natural Science Foundation (No. L202012), the Open Research Project of the State Key Laboratory of Media Convergence and Communication, Communication University of China (No. SKLMCC2020KF008), and the Fundamental Research Funds for the Central Universities (No. 2020RC05). 

\bibliographystyle{gbt7714-numerical}
\bibliography{hudie_reference}

\begin{thebibliography}{18}
\providecommand{\natexlab}[1]{#1}
\providecommand{\url}[1]{#1}
\expandafter\ifx\csname urlstyle\endcsname\relax\else
  \urlstyle{same}\fi
\expandafter\ifx\csname href\endcsname\relax
  \DeclareUrlCommand\doi{\urlstyle{rm}}\else
  \providecommand\doi[1]{\href{https://doi.org/#1}{\nolinkurl{#1}}}\fi

\bibitem[Meng et~al.(2015)Meng, Wang, and Leong]{Meng_2015:SkyStitchAC}
MENG~X, WANG~W, LEONG~B.
\newblock {SkyStitch}: A cooperative multi-{UAV}-based real-time video
  surveillance system with stitching\allowbreak[C]//\allowbreak
Proc. 23rd ACM International Conference on Multimedia (MM).
\newblock New York, USA: ACM, 2015: 261-270.

\bibitem[Balasingam et~al.(2015)Balasingam, Pattipati, Levchuck, and
  Romano]{Balasingam2015:Robust_collaborative}
BALASINGAM~B, PATTIPATI~K, LEVCHUCK~G, et~al.
\newblock Robust collaborative learning by
  multi-agents\allowbreak[C]//\allowbreak
Proc. IEEE Symposium on Computational Intelligence for Security and Defense
  Applications (CISDA).
\newblock Verona, USA: IEEE, 2015: 1-5.

\bibitem[Fei et~al.(2017)Fei, Li, Yang, Xing, Chen, and
  Hanzo]{Fei2017:Wireless_Sensor}
FEI~Z, LI~B, YANG~S, et~al.
\newblock A survey of multi-objective optimization in wireless sensor networks:
  Metrics, algorithms, and open problems\allowbreak[J].
\newblock IEEE Communications Surveys \& Tutorials, 2017, 19\penalty0
  (1):\penalty0 550-586.

\bibitem[Xie et~al.(2009)Xie, Huang, Yang, and Lv]{Xie2009:Multi-Channel}
XIE~X, HUANG~B, YANG~S, et~al.
\newblock Adaptive multi-channel {MAC} protocol for dense {VANET} with
  directional antennas\allowbreak[C]//\allowbreak
Proc. 6th IEEE Consumer Communications and Networking Conference.
\newblock Las Vegas, USA: IEEE, 2009: 1-5.

\bibitem[Lv et~al.(2016)Lv, Gao, Li, Yang, and Hanzo]{Lv2016:Localization}
LV~T, GAO~H, LI~X, et~al.
\newblock Space-time hierarchical-graph based cooperative localization in
  wireless sensor networks\allowbreak[J].
\newblock IEEE Transactions on Signal Processing, 2016, 64\penalty0
  (2):\penalty0 322-334.

\bibitem[Ahlswede et~al.(2000)Ahlswede, Cai, Li, and
  Yeung]{Ahlswede2000:Network_information}
AHLSWEDE~R, CAI~N, LI~S~Y~R, et~al.
\newblock Network information flow\allowbreak[J].
\newblock IEEE Transactions on Information Theory, 2000, 46\penalty0
  (4):\penalty0 1204-1216.

\bibitem[Bassoli et~al.(2013)Bassoli, Marques, Rodriguez, Shum, and
  Tafazolli]{Bassoli2013:netwoek_survey}
BASSOLI~R, MARQUES~H, RODRIGUEZ~J, et~al.
\newblock Network coding theory: A survey\allowbreak[J].
\newblock IEEE Communications Surveys \& Tutorials, 2013, 15\penalty0
  (4):\penalty0 1950-1978.

\bibitem[Katti et~al.(2008)Katti, Rahul, Hu, Katabi, Medard, and
  Crowcroft]{Katti2008:Wireless_Network_Coding}
KATTI~S, RAHUL~H, HU~W, et~al.
\newblock {XORs} in the air: Practical wireless network coding\allowbreak[J].
\newblock IEEE/ACM Transactions on Networking, 2008, 16\penalty0 (3):\penalty0
  497-510.

\bibitem[Cao et~al.(2013)Cao, Lv, Gao, Yang, and Cioffi]{Cao2013:Two-Way_Relay}
CAO~R, LV~T, GAO~H, et~al.
\newblock Achieving full diversity in multi-antenna two-way relay networks via
  symbol-based physical-layer network coding\allowbreak[J].
\newblock IEEE Transactions on Wireless Communications, 2013, 12\penalty0
  (7):\penalty0 3445-3457.

\bibitem[Asterjadhi et~al.(2010)Asterjadhi, Fasolo, Rossi, Widmer, and
  Zorzi]{Asterjadhi2010:coding-based_protocols}
ASTERJADHI~A, FASOLO~E, ROSSI~M, et~al.
\newblock Toward network coding-based protocols for data broadcasting in
  wireless ad hoc networks\allowbreak[J].
\newblock IEEE Transactions on Wireless Communications, 2010, 9\penalty0
  (2):\penalty0 662-673.

\bibitem[Duan et~al.(2015)Duan, Peng, Xu, Zhao, and
  Tian]{Duan2015:all-to-all_broadcasting}
DUAN~P, PENG~L, XU~R, et~al.
\newblock An all-to-all broadcasting protocol using directional antennas in
  multi-hop wireless networks\allowbreak[C]//\allowbreak
Proc. International Conference on Wireless Communications Signal Processing
  (WCSP).
\newblock Nanjing, China: IEEE, 2015: 1-6.

\bibitem[Ma et~al.(2013)Ma, Lin, Zhang, Mao, and
  Vucetic]{Ma2013:Reliability_all-to-all}
MA~L, LIN~Z, ZHANG~Z, et~al.
\newblock Reliability of all-to-all broadcast with network
  coding\allowbreak[C]//\allowbreak
Proc. IEEE Global Communications Conference (GLOBECOM).
\newblock Atlanta, USA: IEEE, 2013: 1991-1996.

\bibitem[Borkotoky et~al.(2019)Borkotoky and
  Pursley]{Borkotoky2019:All-to-All_Broadcast}
BORKOTOKY~S~S, PURSLEY~M~B.
\newblock Adaptive transmission for network-coded all-to-all broadcast in
  multi-hop ad hoc wireless networks\allowbreak[J].
\newblock IEEE Communications Letters, 2019, 23\penalty0 (12):\penalty0
  2412-2416.

\bibitem[Choi et~al.(2009)Choi, Cho, Park, Bae, Chang-Joo, and
  Baik]{Choi_2009:Synchronization}
CHOI~M, CHO~E, PARK~D, et~al.
\newblock A synchronization algorithm of mobile database for ubiquitous
  computing\allowbreak[C]//\allowbreak
Proc. 5th International Joint Conference on INC, IMS, and IDC.
\newblock Seoul, Korea: IEEE, 2009: 416-419.

\bibitem[Son(1988)]{Son1988:Data_Management}
SON~S~H.
\newblock Replicated data management in distributed database
  systems\allowbreak[J].
\newblock SIGMOD Record, 1988, 17\penalty0 (4):\penalty0 62-69.

\bibitem[Paolo(2005)]{Paolo2005:Topology_Control}
PAOLO~S.
\newblock Topology control in wireless ad hoc and sensor
  networks\allowbreak[J].
\newblock ACM Computing Surveys, 2005, 37\penalty0 (2):\penalty0 164-194.

\bibitem[Chen et~al.(2002)Chen, Jamieson, Balakrishnan, and
  Morris]{Chen2002:Topology_Maintenance}
CHEN~B, JAMIESON~K, BALAKRISHNAN~H, et~al.
\newblock Span: An energy-efficient coordination algorithm for topology
  maintenance in ad hoc wireless networks\allowbreak[J].
\newblock Wireless Networks, 2002, 8\penalty0 (5):\penalty0 481-494.

\bibitem[Rajmohan(2002)]{Rajmohan2002:Topology_Control_Survey}
RAJMOHAN~R.
\newblock Topology control and routing in ad hoc networks: A
  survey\allowbreak[J].
\newblock SIGACT News, 2002, 33\penalty0 (2):\penalty0 60-73.

\end{thebibliography}
\newpage
\biographies
\begin{CCJNLbiography}{hudie.eps}{Die Hu}received the B.Eng degree in communication engineering from Wuhan University, China, in Jul. 2018, and the M.Eng degree in system design of flight vehicle from China Academy of Launch Vehicle Technology in Jun. 2021. She was also a visiting student at Beijing University of Posts and Telecommunications from Jan. 2020 to Apr. 2021. She is currently a research engineer at China Academy of Launch Vehicle Technology. Her research interests include distributed information system, flying ad hoc networks and routing protocols.
\end{CCJNLbiography}

\begin{CCJNLbiography}{zhuxuejun.eps}{Xuejun Zhu}received the B.Eng degree in automatic control from National University of Defense Technology, China, in 1984, and the M.Eng degree in aircraft navigation and control from China Academy of Launch Vehicle Technology in 1987. She is currently a chief designer and professorial research fellow with China Academy of Launch Vehicle Technology. In 2019, she was elected to the Academician of Chinese Academy of Sciences. Her research interests include networking technologies, flight vehicle design and electrical system design.
\end{CCJNLbiography}

\begin{CCJNLbiography}{gongmin.eps}{Min Gong} received the B.Eng degree in information engineering from Beijing Institute of Technology, China, in 2005, and the PhD degree in information and communication engineering from Tsinghua University, China, in 2010. He is currently a chief designer with China Academy of Launch Vehicle Technology. His research interests include mobile ad hoc networks, distributed information processing, flight vehicle design and game theory.
\end{CCJNLbiography}

\begin{CCJNLbiography}{Shaoshi_Yang.eps}{Shaoshi Yang}received the B.Eng degree in information engineering from Beijing University of Posts and Telecommunications (BUPT), China, in 2006, and the PhD degree in electronics and electrical engineering from University of Southampton, U.K., in 2013. From 2008 to 2009, he was a researcher of WiMAX standardization with Intel Labs China. From 2013 to 2016, he was a Research Fellow with the School of Electronics and Computer Science, University of Southampton. From 2016 to 2018, he was a Principal Engineer with Huawei Technologies Co. Ltd., where he made significant contributions to the company’s products and solutions associated with 5G base stations, wideband IoT, and cloud gaming/VR. He is currently a Full Professor with BUPT. His research expertise includes 5G wireless networks, massive MIMO, iterative detection and decoding, mobile ad hoc networks, distributed artificial intelligence, and cloud gaming/VR. He is a member of the Isaac Newton Institute for Mathematical Sciences, Cambridge University, and a Senior Member of IEEE. He received the Dean’s Award for Early Career Research Excellence from  University of Southampton in 2015, the Huawei President Award of Wireless Innovations in 2018, the IEEE Technical Committee on Green Communications \& Computing (TCGCC) Best Journal Paper Award in 2019, and the IEEE Communications Society Best Survey Paper Award in 2020. He is an Editor for IEEE Systems Journal, IEEE Wireless Communications Letters, and Signal Processing (Elsevier). He was also an invited international reviewer of the Austrian Science Fund (FWF).
\end{CCJNLbiography}

\end{document}